\documentclass[10pt,reqno]{amsart}
\usepackage[hypertex]{hyperref}

\usepackage{amsmath,amsthm}
\usepackage{amssymb}
\usepackage{euscript}
\usepackage[frame,ps,matrix,arrow,curve,rotate,all,2cell,tips]{xy}

\numberwithin{equation}{subsection}

\newcommand{\sqsp}{\renewcommand{\baselinestretch}{1.1}\tiny\normalsize}

\newcommand{\relphantom[1]}{\mathrel{\phantom{#1}}}
\newtheorem{theorem}[subsection]{Theorem}

\newtheorem{proposition}[subsection]{Proposition}
\newtheorem{corollary}[subsection]{Corollary}

\theoremstyle{definition}

\newcommand{\br}{\mathfrak{B}_n}
\newcommand{\fg}{\mathfrak{g}}
\newcommand{\fh}{\mathfrak{h}}
\newcommand{\fgl}{\mathfrak{gl}}
\newcommand{\fsl}{\mathfrak{sl}}
\newcommand{\hei}{\mathsf{H}}
\newcommand{\bk}{\mathbf{k}}
\newcommand{\bC}{\mathbf{C}}
\newcommand{\bR}{\mathbf{R}}
\newcommand{\poincare}{\fsl(2)^*}

\DeclareMathOperator{\Hom}{Hom}
\DeclareMathOperator{\Aut}{Aut}

\def\psimatrix{{\begin{pmatrix}q^{-1} & 0 & 0 & 0\\ 0 & q^{-1} - q & 1 & 0\\ 0 & 1 & 0 & 0\\ 0 & 0 & 0 & q^{-1}\end{pmatrix}}}

\def\matrixb{{\begin{pmatrix}0 & b\\ 0 & 0 \end{pmatrix}}}

\def\matrixc{{\begin{pmatrix}0 & 0\\ c & 0 \end{pmatrix}}}

\def\matrixad{{\begin{pmatrix}a & 0\\ 0 & d \end{pmatrix}}}

\def\Phib{{\begin{pmatrix}0 & 0 & 0 & b^2\lambda \\
0 & 0 & 0 & 0 \\
0 & 0 & 0 & 0 \\
0 & 0 & 0 & 0 \\ \end{pmatrix}}}

\def\Phic{{\begin{pmatrix}0 & 0 & 0 & 0 \\
0 & 0 & 0 & 0 \\
0 & 0 & 0 & 0 \\
c^2\lambda & 0 & 0 & 0 \\\end{pmatrix}}}

\def\Phiad{{q\lambda\begin{pmatrix}a^2q^{-1} & 0 & 0 & 0\\
0 & ad\beta & ad & 0\\
0 & ad & 0 & 0\\
0 & 0 & 0 & d^2q^{-1}
\end{pmatrix}}}

\def\alphaone{{\begin{pmatrix}0 & 0 & 0 \\
0 & 0 & 0 \\
a_{31} & 0 & 0 \\ \end{pmatrix}}}

\def\alphatwo{{\begin{pmatrix}0 & 0 & 0 \\
a_{21} & 0 & 0 \\
0 & a_{32} & 0 \\ \end{pmatrix}}}

\def\alphathree{{\begin{pmatrix}0 & 0 & 0 \\
a_{21} & 0 & 0 \\
0 & 0 & a_{33} \\ \end{pmatrix}}}

\def\alphafour{{\begin{pmatrix}a_{11} & 0 & 0 \\
0 & a_{22} & 0 \\
0 & 0 & a_{33} \\ \end{pmatrix}}}

\def\alphafive{{\begin{pmatrix}a_{11} & 0 & 0 \\
0 & 0 & 0 \\
0 & a_{32} & 0 \\ \end{pmatrix}}}

\def\alphasix{{\begin{pmatrix}a_{11} & 0 & 0 \\
0 & 0 & a_{23} \\
0 & 0 & 0 \\ \end{pmatrix}}}

\def\alphaeight{{\begin{pmatrix}0 & a_{12} & 0 \\
0 & 0 & a_{23} \\
0 & 0 & 0 \\ \end{pmatrix}}}

\def\alphanine{{\begin{pmatrix}0 & a_{12} & 0 \\
0 & 0 & 0 \\
0 & 0 & a_{33} \\ \end{pmatrix}}}

\def\alphaten{{\begin{pmatrix}0 & 0 & a_{13} \\
0 & 0 & 0 \\
0 & 0 & 0 \\ \end{pmatrix}}}

\def\matrixH{{\begin{pmatrix}a_{22}a_{33} - a_{23}a_{32} & a_{12} & a_{13}\\
0 & a_{22} & a_{23}\\
0 & a_{32} & a_{33}\end{pmatrix}}}

\def\alphaHa{{\begin{pmatrix}1 & 0 & 0\\
0 & a & 0\\
0 & 0 & a^{-1}\end{pmatrix}}}

\def\alphaHb{{\begin{pmatrix}1 & 0 & 0\\
0 & b & 0\\
0 & 0 & b^{-1}\end{pmatrix}}}

\def\alphaPa{{\begin{pmatrix}1 & 0 & 0\\
a_{21} & a_{22} & a_{23}\\
a_{31} & a_{32} & a_{33}\end{pmatrix}}}

\def\alphaPb{{\begin{pmatrix}a_{11} & 0 & 0\\
a_{21} & 0 & 0\\
a_{31} & 0 & 0\end{pmatrix}}}

\def\alphaSa{{\begin{pmatrix}1 & c & a\\
-2ab & b & -a^2b\\
-2b^{-1}c & -b^{-1}c^2 & b^{-1}\end{pmatrix}}}

\def\alphaSb{{\begin{pmatrix}-1 & c & a\\
2b^{-1}c & -b^{-1}c^2 & b^{-1}\\
2ab & b & -a^2b\end{pmatrix}}}

\def\alphaSc{{\begin{pmatrix}c & a & \dfrac{1-c^2}{4a}\\
b & \dfrac{ab}{c-1} & \dfrac{b(1-c)}{4a}\\
\dfrac{1-c^2}{b} & \dfrac{a(1-c)}{b} & \dfrac{(c^2-1)(c+1)}{4ab}\end{pmatrix}}}

\def\alphaonea{{\begin{pmatrix}1 & 0 & 0\\
0 & a_{22} & 0\\
0 & 0 & a_{22}^{-1}\end{pmatrix}}}

\def\alphaoneb{{\begin{pmatrix}-1 & 0 & 0\\
0 & 0 & a_{32}^{-1}\\
0 & a_{32} & 0\end{pmatrix}}}

\def\alphatwoa{{\begin{pmatrix}-1 & 0 & a_{13}\\
0 & 0 & a_{32}^{-1}\\
2a_{13}a_{32} & a_{32} & -a_{13}^2a_{32}\end{pmatrix}}}

\def\alphatwob{{\begin{pmatrix}1 & 0 & a_{13}\\
-2a_{13}a_{22} & a_{22} & -a_{13}^2a_{22}\\
0 & 0 & a_{22}^{-1}\end{pmatrix}}}

\def\alphathreea{{\begin{pmatrix}1 & a_{12} & 0\\
0 & a_{22} & 0\\
-2a_{12}a_{22}^{-1} & -a_{12}^2a_{22}^{-1} & a_{22}^{-1}\end{pmatrix}}}

\def\alphathreeb{{\begin{pmatrix}-1 & a_{12} & 0\\
2a_{12}a_{32}^{-1} & -a_{12}^2a_{32}^{-1} & a_{32}^{-1}\\
0 & a_{32} & 0\end{pmatrix}}}

\begin{document}

\title{The Hom-Yang-Baxter equation and Hom-Lie algebras}
\author{Donald Yau}

\begin{abstract}
Motivated by recent work on Hom-Lie algebras, a twisted version of the Yang-Baxter equation, called the Hom-Yang-Baxter equation (HYBE), was introduced by the author in \cite{yau5}.   In this paper, several more classes of solutions of the HYBE are constructed.  Some of these solutions of the HYBE are closely related to the quantum enveloping algebra of $\fsl(2)$, the Jones-Conway polynomial, and Yetter-Drinfel'd modules.  Under some invertibility conditions, we construct a new infinite sequence of solutions of the HYBE from a given one.
\end{abstract}



\address{Department of Mathematics\\
    The Ohio State University at Newark\\
    1179 University Drive\\
    Newark, OH 43055, USA}
\email{dyau@math.ohio-state.edu}

\date{\today}
\maketitle

\sqsp

%%%%%%%%%%%%%%%%%%%%%%
\section{Introduction}
%%%%%%%%%%%%%%%%%%%%%%

The purpose of this paper is to construct some concrete classes of solutions of the Hom-Yang-Baxter equation (HYBE), which was introduced by the author in \cite{yau5}.  Let us first recall some motivations for the HYBE.  In \cite{hls} a generalization of Lie algebras, called Hom-Lie algebras, were introduced in which the Jacobi identity is twisted by a linear self-map.  More precisely, a \emph{Hom-Lie algebra} $L = (L,[-,-],\alpha)$ consists of a vector space $L$, a bilinear skew-symmetric bracket $[-,-] \colon L \otimes L \to L$, and a linear self-map $\alpha \colon L \to L$ such that the following \emph{Hom-Jacobi identity} holds:
\begin{equation}
\label{eq:hom-jacobi}
[[x,y],\alpha(z)] + [[z,x],\alpha(y)] + [[y,z],\alpha(x)]] = 0.
\end{equation}
If, moreover, it satisfies
\[
\alpha[x,y] = [\alpha(x),\alpha(y)]
\]
for $x,y \in L$, then it is said to be \emph{multiplicative}.  A Lie algebra can be regarded as a multiplicative Hom-Lie algebra with $\alpha = Id$.  Hom-Lie algebras (without multiplicativity) were introduced in \cite{hls} to describe the structures on some $q$-deformations of the Witt and the Virasoro algebras.  They are also closely related to discrete and deformed vector fields, differential calculus \cite{hls,ls,ls2,rs,ss}, and number theory \cite{larsson}.  Earlier precursors of Hom-Lie algebras can be found in \cite{as,hu,liu}.  The reader is referred to \cite{mak2,ms}, \cite{yau}-\cite{yau12}, and the references therein for discussions about other Hom-type structures.

Recall that the Yang-Baxter equation (YBE) states
\begin{equation}
\label{eq:YBE}
(Id_V \otimes B) \circ (B \otimes Id_V) \circ (Id_V \otimes B) = (B \otimes Id_V) \circ (Id_V \otimes B) \circ (B \otimes Id_V),
\end{equation}
where $V$ is a vector space and $B \colon V^{\otimes 2} \to V^{\otimes 2}$ is a bilinear automorphism.  A solution of the YBE is also called an \emph{$R$-matrix} on $V$.  The YBE was first introduced in the context of statistical mechanics \cite{baxter,baxter2,yang}.  It plays an important role in many topics in mathematical physics, including quantum groups, quantum integrable systems \cite{fad}, braided categories \cite{js1,js2,js3}, the Zamolodchikov tetrahedron equation in higher-dimensional categories \cite{bc,kv}, and invariants of knots and links \cite{tur,yetter}, among others.  Many $R$-matrices are known.  In particular, various classes of quantum groups were introduced precisely for the purpose of constructing $R$-matrices \cite{dri85,dri87,dri89,hay,kassel,lt,majid91,sch}.  Among the many classes of $R$-matrices, a particularly interesting class comes from Lie algebras \cite{bc}.  Thinking of Hom-Lie algebras as $\alpha$-twisted analogues of Lie algebras, this raises the question: What is the $\alpha$-twisted analogue of the YBE that corresponds to Hom-Lie algebras?  A natural answer, given in \cite{yau5}, is the \emph{Hom-Yang-Baxter equation} (HYBE):
\begin{equation}
\label{eq:HYBE}
(\alpha \otimes B) \circ (B \otimes \alpha) \circ (\alpha \otimes B) = (B \otimes \alpha) \circ (\alpha \otimes B) \circ (B \otimes \alpha),
\end{equation}
where $V$ is a vector space, $\alpha \colon V \to V$ is a linear map, and $B \colon V^{\otimes 2} \to V^{\otimes 2}$ is a bilinear (not-necessarily invertible) map that commutes with $\alpha^{\otimes 2}$.  In this case, we say that $B$ is a solution of the HYBE for $(V,\alpha)$.  The YBE can be regarded as the special case of the HYBE in which $\alpha = Id$ and $B$ is invertible.

As in the case of $R$-matrices, solutions of the HYBE can be extended to operators that satisfy the braid relations.  With an additional invertibility condition, we obtain a representation of the braid group.  Indeed, suppose $B$ is a solution of the HYBE for $(V,\alpha)$ \eqref{eq:HYBE}.  For integers $n \geq 3$ and $1 \leq i \leq n-1$, define the operators $B_i \colon V^{\otimes n} \to V^{\otimes n}$ by
\begin{equation}
\label{eq:Bi}
B_i = \begin{cases} B \otimes \alpha^{\otimes (n-2)} & \text{ if $i = 1$},\\
\alpha^{\otimes(i-1)} \otimes B \otimes \alpha^{\otimes (n - i - 1)} & \text{ if $1 < i < n-1$}, \\
\alpha^{\otimes(n-2)} \otimes B & \text{ if $i = n-1$}.\end{cases}
\end{equation}
It is proved in \cite[Theorem 1.4]{yau5} that the maps $B_i$ ($1 \leq i \leq n - 1$) satisfy the braid relations
\begin{equation}
\label{eq:braidrelations}
\begin{split}
B_i B_j &= B_j B_i \text{ if $|i - j| > 1$},\\
B_iB_{i+1}B_i &= B_{i+1}B_iB_{i+1}.
\end{split}
\end{equation}
Moreover, let $\br$ be the braid group on $n$ strands and $\sigma_i \in \br$ be the element representing crossing the $i$th and the $(i+1)$st strands with the former going under the latter.  The elements $\sigma_i$ for $i \in \{1,\ldots,n-1\}$ generate $\br$ and satisfy the defining braid relations \eqref{eq:braidrelations} \cite{artin2,artin}.  If both $\alpha$ and $B$ are invertible, then clearly so are the $B_i$ \eqref{eq:Bi}.  In this case, there is a unique group morphism $\rho^B_n \colon \br \to \Aut(V^{\otimes n})$ satisfying
\[
\rho^B_n(\sigma_i) = B_i.
\]
This generalizes the usual braid group representations associated to $R$-matrices, as discussed, for example, in \cite[X.6.2]{kassel}.

From the discussion above, it is reasonable to say that it is an important task to construct solutions of the HYBE.  Several classes of solutions of the HYBE were constructed in \cite{yau5}.  The purpose of this paper is to construct several more interesting classes of such solutions.

The rest of this paper is organized as follows.  In the next section, we construct solutions of the HYBE from $R$-matrices.  As examples, we apply this result to the $R$-matrices on the simple two-dimensional $U_q(\fsl(2))$-module and their higher dimensional analogues.  In section \ref{sec:HL}, we construct classes of solutions of the HYBE from multiplicative Hom-Lie algebras.  In particular, we compute all the Lie algebra morphisms on the Heisenberg algebra, the $(1+1)$-Poincar\'{e} algebra, and $\fsl(2)$, and describe their associated solutions of the HYBE.  In section \ref{sec:YD}, we construct solutions of the HYBE from morphisms on Yetter-Drinfel'd modules.  In section \ref{sec:Vn}, starting from an arbitrary solution of the HYBE, we construct an infinite sequence of new solutions of the HYBE using the Iwahori map and a result from \cite{yau5}.

%%%%%%%%%%%%%%%%%%%%%%%%%%%%%%%%%%%%%%%%%%%%%%%%%%%%%%%%%%
\section{Twisting $R$-matrices into solutions of the HYBE}
\label{sec:twist}

Thinking of the HYBE \eqref{eq:HYBE} as an $\alpha$-twisted version of the YBE \eqref{eq:YBE}, it makes sense that one should be able to twist $R$-matrices into solutions of the HYBE.  We can derive this twisting construction of solutions of the HYBE from the following result.

\begin{proposition}
\label{prop:twist}
Let $B$ be a solution of the HYBE \eqref{eq:HYBE} for $(V,\alpha)$, and let $\beta \colon V \to V$ be a linear map such that $\beta^{\otimes 2} \circ B = B \circ \beta^{\otimes 2}$ and $\beta \circ \alpha = \alpha \circ \beta$.  Then
\[
B_\beta  = \beta^{\otimes 2} \circ  B \colon V^{\otimes 2} \to V^{\otimes 2}
\]
is a solution of the HYBE for $(V, \beta \circ \alpha)$.
\end{proposition}

\begin{proof}
Let us write $\gamma = \beta \circ \alpha$ in this proof.  The two assumptions about $\beta$ imply
\[
B_\beta \circ \gamma^{\otimes 2} = \gamma^{\otimes 2} \circ B_\beta.
\]
To check that $B_\beta$ satisfies the HYBE for $(V,\gamma)$, observe that
\[
B_\beta \otimes \gamma = (B \otimes \alpha) \circ \beta^{\otimes 3} = \beta^{\otimes 3} \circ (B \otimes \alpha)
\]
and
\[
\gamma \otimes B_\beta = (\alpha \otimes B) \circ \beta^{\otimes 3} = \beta^{\otimes 3} \circ (\alpha \otimes B).
\]
Therefore, we obtain the desired HYBE,
\[
(\gamma \otimes B_\beta) \circ (B_\beta \otimes \gamma) \circ (\gamma \otimes B_\beta) = (B_\beta \otimes \gamma) \circ (\gamma \otimes B_\beta) \circ (B_\beta \otimes \gamma),
\]
by applying $\beta^{\otimes 3} \circ \beta^{\otimes 3} \circ \beta^{\otimes 3}$ to the HYBE \eqref{eq:HYBE} for $(V,\alpha)$.
\end{proof}

Two special cases of Proposition \ref{prop:twist} follow.  The following result says that each solution of the HYBE gives rise to another solution of the HYBE.

\begin{corollary}
\label{cor1:twist}
Let $B$ be a solution of the HYBE \eqref{eq:HYBE} for $(V,\alpha)$.  Then
\[
B_\alpha = \alpha^{\otimes 2} \circ B
\]
is a solution of the HYBE for $(V,\alpha^2)$.
\end{corollary}

\begin{proof}
This is the special case of Proposition \ref{prop:twist} when $\beta = \alpha$.
\end{proof}

The next result says that solutions of the YBE can be twisted along compatible morphisms to yield solutions of the HYBE.

\begin{corollary}
\label{cor2:twist}
Let $V$ be a vector space, $B \colon V^{\otimes 2} \to V^{\otimes 2}$ be a bilinear (not-necessarily invertible) map that satisfies the YBE \eqref{eq:YBE}, and $\beta \colon V \to V$ be a linear self-map such that
\begin{equation}
\label{eq:compatibility}
\beta^{\otimes 2} \circ B = B \circ \beta^{\otimes 2}.
\end{equation}
Then
\[
B_\beta = \beta^{\otimes 2} \circ B \colon V^{\otimes 2} \to V^{\otimes 2}
\]
is a solution of the HYBE \eqref{eq:HYBE} for $(V,\beta)$.
\end{corollary}

\begin{proof}
This is the special case of Proposition \ref{prop:twist} when $\alpha$ is the identity map on $V$.
\end{proof}

In order to apply Corollary \ref{cor2:twist}, we take some important $R$-matrices $B$ and compute \emph{all} the linear maps $\alpha$ that satisfy the compatibility condition \eqref{eq:compatibility}.  Here we work over the field $\bC$ of complex numbers.

As the first example, consider the (unique up to isomorphism) two-dimensional simple $U_q(\fsl(2))$-module $V_1$.  Here $q \in \bC \smallsetminus\{0,\pm 1\}$ and $U_q(\fsl(2))$ is the quantum enveloping algebra of the Lie algebra $\fsl(2)$ \cite{dri85,dri87,jimbo,kassel,kr,majid}.  We do not need to know the structure of $U_q(\fsl(2))$ for the discussion below.  There is a basis $\{v_0,v_1\}$ of $V_1$ in which $v_0$ is a highest weight vector.  With respect to the basis $\{v_0 \otimes v_0, v_0 \otimes v_1, v_1 \otimes v_0, v_1 \otimes v_1\}$ of $V_1^{\otimes 2}$, the only interesting $U_q(\fsl(2))$-linear $R$-matrices on $V_1$ are
\begin{equation}
\label{eq:Phi}
\Phi_{q,\lambda} = q\lambda\psimatrix
\end{equation}
for any $\lambda \in \bC \smallsetminus\{0\}$.  The only other $U_q(\fsl(2))$-linear $R$-matrices on $V_1$ are 
\begin{enumerate}
\item
the non-zero scalar multiples of $Id_{V_1}^{\otimes 2}$ and 
\item
matrices obtained from $\Phi_{q,\lambda}$ with a change of basis and a switch of $q$ and $q^{-1}$.  
\end{enumerate}
A clear exposition of these $R$-matrices on $V_1$ can be found in \cite[VIII.1]{kassel}.  The following result describes all the linear maps that are compatible with $\Phi_{q,\lambda}$ and their induced solutions of the HYBE.

\begin{theorem}
\label{thm:V1}
With respect to the basis $\{v_0,v_1\}$ of the two-dimensional simple $U_q(\fsl(2))$-module $V_1$, a linear map $\alpha \colon V_1 \to V_1$ satisfies \eqref{eq:compatibility} with $B = \Phi_{q,\lambda}$ \eqref{eq:Phi} if and only if
\begin{equation}
\label{eq:alpha2}
\alpha = \matrixb,\quad \matrixc,\quad\text{or}\quad \matrixad
\end{equation}
for any scalars $a,b,c$, $d \in \bC$.  Their corresponding solutions
\[
\Phi_\alpha = \alpha^{\otimes 2} \circ \Phi_{q,\lambda}
\]
of the HYBE \eqref{eq:HYBE} for $(V_1,\alpha)$ are
\[
\Phib, \quad \Phic, \quad\text{and}\quad \Phiad,
\]
respectively, where $\beta = q^{-1} - q$.  The last type of $\Phi_\alpha$ is invertible if and only if $ad \not= 0$.
\end{theorem}

Theorem \ref{thm:V1} is actually a special case of the next result, so we will only prove the next result.  There are important higher dimensional versions of the $R$-matrices $\Phi_{q,\lambda}$ \eqref{eq:Phi}.  Suppose $N \geq 2$ and $V$ is an $N$-dimensional vector space with a basis $\{e_1,\ldots,e_N\}$.  Fix non-zero scalars $q \not= \pm 1$ and $\lambda$ in $\bC$.  Consider the bilinear map $B_{q,\lambda} \colon V^{\otimes 2} \to V^{\otimes 2}$ defined by
\begin{equation}
\label{eq:Bqlambda}
B_{q,\lambda}(e_i \otimes e_j) =
\begin{cases}
\lambda qe_i \otimes e_i & \text{ if $i = j$},\\
\lambda e_j \otimes e_i & \text{ if $i < j$},\\
\lambda e_j \otimes e_i + \lambda(q - q^{-1})e_i \otimes e_j & \text{ if $i > j$}.
\end{cases}
\end{equation}
It is known that $B_{q,\lambda}$ is an $R$-matrix on $V$ \cite[VIII.1.4]{kassel}.  The $R$-matrix $B_{q^{-1},q}$ is known as the \emph{Jimbo operator of type} $A^{(1)}_{N-1}$ \cite{jimbo,jimbo1,jimbo2}.  Switching $q$ and $q^{-1}$, the $R$-matrix $B_{q,q^{-1}}$ is used in constructing the quantum exterior algebra \cite[section 3]{jr}.  Moreover, the case $N=2$ contains the $R$-matrix $\Phi_{q,\lambda}$ \eqref{eq:Phi}.  In fact, if we switch the basis vectors $v_0$ and $v_1$ in $V_1$, then
\[
\Phi_{q,\lambda} = B_{q^{-1},q\lambda}.
\]
The general case of the $R$-matrix $B_{q,\lambda}$ is a crucial ingredient in establishing the existence of the polynomial invariant of links known as the \emph{Jones-Conway polynomial} \cite{conway,homfly,jones1,jones2}, as discussed in \cite[X.4 and XII.5]{kassel}.

The following result, which generalizes Theorem ~\ref{thm:V1} to higher dimensions, describes all the linear maps that are compatible with $B_{q,\lambda}$ and their induced solutions of the HYBE.

\begin{theorem}
\label{thm:Bqlambda}
With respect to the basis $\{e_1,\ldots,e_N\}$ of $V$, a linear map $\alpha \colon V \to V$ satisfies \eqref{eq:compatibility} with $B = B_{q,\lambda}$ \eqref{eq:Bqlambda} if and only if $\alpha$ satisfies the following two conditions:
\begin{enumerate}
\item
Each column in the matrix $(a_{ij})$ for $\alpha$ has at most one non-zero entry.
\item
If $1 \leq i < j \leq N$ and $a_{ki}a_{lj} \not= 0$, then $k < l$.
\end{enumerate}
Moreover, suppose that $\alpha$ satisfies these two conditions and that $a_i$ denotes either $0$ or the only non-zero entry $a_{k(i),i}$ in the $i$th column in $(a_{ij})$, if it exists. Then the solution
\[
B_\alpha = \alpha^{\otimes 2} \circ B_{q,\lambda}
\]
of the HYBE for $(V,\alpha)$ \eqref{eq:HYBE} induced by $\alpha$ is given by
\begin{equation}
\label{eq:BalphaN}
B_\alpha(e_i \otimes e_j) =
\begin{cases}
\lambda qa_i^2 \left(e_{k(i)} \otimes e_{k(i)}\right) & \text{ if $i = j$},\\
\lambda a_ia_j \left(e_{k(j)} \otimes e_{k(i)}\right) & \text{ if $i < j$},\\
\lambda a_ia_j\left(e_{k(j)} \otimes e_{k(i)} + (q - q^{-1})e_{k(i)} \otimes e_{k(j)}\right) & \text{ if $i > j$}.
\end{cases}
\end{equation}
\end{theorem}

For example, when $N = 2$, the linear maps $\alpha$ that satisfy the two conditions in Theorem ~\ref{thm:Bqlambda} are listed in \eqref{eq:alpha2}.  When $N = 3$, the linear maps $\alpha$ that satisfy the two conditions in Theorem ~\ref{thm:Bqlambda} are listed below, where $a_{ij} \in \bC$:
\[
\begin{tabular}{c}
$\alphaone, \alphatwo, \alphathree, \alphafour, \alphafive$, \\
$\alphasix, \alphaeight, \alphanine, \alphaten$.
\end{tabular}
\]
Note also that if $\alpha$ satisfies the two conditions in Theorem ~\ref{thm:Bqlambda}, then $\alpha$ is invertible if and only if its matrix is a diagonal matrix with non-zero diagonal entries.  In this case, the induced solution $B_\alpha$ \eqref{eq:BalphaN} of the HYBE for $(V,\alpha)$ is also invertible, since $B_{q,\lambda}$ \eqref{eq:Bqlambda} is invertible.

\begin{proof}[Proof of Theorem ~\ref{thm:Bqlambda}]
It remains to prove the first part of Theorem ~\ref{thm:Bqlambda}, which is a classification of the linear maps that are compatible (in the sense of ~\eqref{eq:compatibility}) with $B_{q,\lambda}$ ~\eqref{eq:Bqlambda}.  The second part, regarding the induced solutions $B_\alpha = \alpha^{\otimes 2} \circ B_{q,\lambda}$ of the HYBE, follows from the first part, the definition of $B_{q,\lambda}$, and Corollary \ref{cor2:twist}.

To compute the desired maps $\alpha$, first note that
\[
B_{q,\lambda} = \lambda B_{q,1}.
\]
Since $\lambda \not= 0$, it follows that $\alpha^{\otimes 2}$ commutes with $B_{q,\lambda}$ if and only if it commutes with $B_{q,1}$.  Therefore, we only have to consider the special case $B_{q,1}$.

First assume that $\alpha^{\otimes 2}$ commutes with $B_{q,1}$.  We must show that the two conditions in Theorem ~\ref{thm:Bqlambda} hold.  Write $(a_{ij})$ for the matrix representing $\alpha$ with respect to the basis $\{e_1, \ldots , e_N\}$.  For each $i \in \{1, \ldots , N\}$, we have
\begin{equation}
\label{eq:alphaBii}
\begin{split}
(\alpha^{\otimes 2} \circ B_{q,1})(e_i \otimes e_i)
&= q\alpha(e_i) \otimes \alpha(e_i) \\
&= \sum_{k,j} (qa_{ji}a_{ki})e_k \otimes e_j,
\end{split}
\end{equation}
where both $k$ and $j$ run through $1$ to $N$.  On the other hand, we have
\begin{equation}
\label{eq:Balphaii}
\begin{split}
&(B_{q,1} \circ \alpha^{\otimes 2})(e_i \otimes e_i)\\
&= B_{q,1}\Bigl(\sum_{j,k} (a_{ji}a_{ki})e_j \otimes e_k\Bigr)\\
&= \sum_j (qa_{ji}^2)e_j \otimes e_j + \sum_{j<k}(a_{ji}a_{ki})e_k \otimes e_j
 + \sum_{j>k} (a_{ji}a_{ki})(e_k \otimes e_j + (q-q^{-1})e_j\otimes e_k)\\
&= \sum_j (qa_{ji}^2)e_j \otimes e_j + \sum_{k<j} (a_{ji}a_{ki})e_k \otimes e_j
 + \sum_{j<k} \left(a_{ji}a_{ki}(1 + q - q^{-1})\right)e_k \otimes e_j.
\end{split}
\end{equation}
Since $q \not= 1$, we infer from \eqref{eq:alphaBii} and \eqref{eq:Balphaii} that
\[
a_{ji}a_{ki} = 0 \quad \text{for $1 \leq i \leq N$ and $k<j$}.
\]
In other words, each column in the matrix $(a_{ij})$ for $\alpha$ has at most one non-zero entry.  This shows that the first condition in Theorem ~\ref{thm:Bqlambda} is necessary in order for $\alpha^{\otimes 2}$ to commute with $B_{q,1}$.

Now we show that the second condition is also necessary.  So assume that $a_{ki}a_{lj} \not= 0$ for some $i < j$.  We must show that $k < l$.  From the previous paragraph, we have
\[
\alpha(e_i) = a_{ki}e_k\quad\text{and}\quad
\alpha(e_j) = a_{lj}e_l.
\]
Since $i < j$, we have
\begin{equation}
\label{eq:alphaBij}
\begin{split}
(\alpha^{\otimes 2} \circ B_{q,1})(e_i \otimes e_j)
&= \alpha(e_j) \otimes \alpha(e_i) \\
&= (a_{ki}a_{lj})e_l \otimes e_k.
\end{split}
\end{equation}
On the other hand, we have
\begin{equation}
\label{eq:Balphaij}
\begin{split}
(B_{q,1} \circ \alpha^{\otimes 2})(e_i \otimes e_j) &= (a_{ki}a_{lj}) B_{q,1}(e_k \otimes e_l)\\
&= \begin{cases}
(a_{ki}a_{lj})q e_k \otimes e_k & \text{ if $k=l$},\\
(a_{ki}a_{lj})e_l \otimes e_k & \text{ if $k<l$},\\
(a_{ki}a_{lj})\left(e_l\otimes e_k + (q-q^{-1})e_k\otimes e_l\right) & \text{ if $k>l$}.
\end{cases}
\end{split}
\end{equation}
Since $q^2 \not= 1$, we infer from \eqref{eq:alphaBij} and \eqref{eq:Balphaij} that, if $k \geq l$, then $a_{ki}a_{lj} = 0$, which contradicts the assumption $a_{ki}a_{lj} \not= 0$.  Therefore, we must have $k<l$, proving the necessity of the second condition in Theorem ~\ref{thm:Bqlambda}.

We have shown that a linear map $\alpha \colon V \to V$ for which $\alpha^{\otimes 2}$ commutes with $B_{q,1}$ \eqref{eq:Bqlambda} must satisfy the two conditions in Theorem ~\ref{thm:Bqlambda}.  Conversely, if the matrix of $\alpha$ satisfies those two conditions, then a direct computation (most of which is already shown above) shows that $\alpha^{\otimes 2} \circ B_{q,1}$ and $B_{q,1} \circ \alpha^{\otimes 2}$ are equal when applied to $e_i \otimes e_j$ for $i = j$, $i < j$, or $i > j$.  In other words, those two conditions are both necessary and sufficient in order for $\alpha^{\otimes 2}$ to commute with $B_{q,1}$, and hence also with $B_{q,\lambda}$ in general.
\end{proof}

%%%%%%%%%%%%%%%%%%%%%%%%%%%%%%%%%%%%%%%%%%%%%%%%%%%%%
\section{Solutions of the HYBE from Hom-Lie algebras}
\label{sec:HL}

Another way to generate solutions of the HYBE is to generalize the $R$-matrices associated to Lie algebras \cite{bc}.  Suppose $L = (L,[-,-],\alpha)$ is a multiplicative Hom-Lie algebra, as defined in the first paragraph of the Introduction.  Set
\[
L' = \bC \oplus L
\]
with
\[
\alpha(a,x) = (a,\alpha(x))
\]
for $a \in \bC$ and $x \in L$.  It is shown in \cite[Theorem 1.1]{yau5} that there is a solution
\begin{equation}
\label{eq:BalphaL}
B_\alpha((a,x) \otimes (b,y)) = (b,\alpha(y)) \otimes (a,\alpha(x)) + (1,0) \otimes (0,[x,y]).
\end{equation}
of the HYBE \eqref{eq:HYBE} for $(L',\alpha)$.  Moreover, if $\alpha$ is invertible, then so is $B_\alpha$ \eqref{eq:BalphaL} \cite[Corollary 3.3]{yau5}, where
\[
B_\alpha^{-1}((a,x) \otimes (b,y)) = (b,\alpha^{-1}(y)) \otimes (a,\alpha^{-1}(x)) + (0,\alpha^{-2}[x,y]) \otimes (1,0).
\]
The operator $B_\alpha^{-1}$ is a solution of the HYBE for $(L',\alpha^{-1})$ \cite[Proposition 2.6]{yau5}.  Furthermore, in this case we have a braid group representation given by $\sigma_i \mapsto B_i$ \eqref{eq:Bi}, as discussed in the Introduction.  These observations are useful as long as we can produce concrete examples of multiplicative Hom-Lie algebras.

One systematic method for constructing multiplicative Hom-Lie algebras goes like this: Let $(\fg,[-,-])$ be a Lie algebra and $\alpha \colon \fg \to \fg$ be a Lie algebra morphism.  Then
\[
\fg_\alpha = (\fg,[-,-]_\alpha,\alpha)
\]
is a multiplicative Hom-Lie algebra \cite{yau2}, where $[-,-]_\alpha = \alpha \circ [-,-]$.  In fact, the Hom-Jacobi identity \eqref{eq:hom-jacobi} for $[-,-]_\alpha$ is $\alpha^2$ applied to the Jacobi identity of $[-,-]$.  To use this recipe to construct Hom-Lie algebras, one needs to compute Lie algebra morphisms on a given Lie algebra $\fg$.  Later in this section, we compute \emph{all} the Lie algebra morphisms on some particularly important Lie algebras and describe their induced Hom-Lie brackets $[-,-]_\alpha$.  One then obtains concrete solutions $B_\alpha$ \eqref{eq:BalphaL} of the HYBE and, when $\alpha$ is invertible, braid group representations, as discussed in the previous paragraph.

First we consider when two multiplicative Hom-Lie algebras of the form
\[
L_\alpha = (L,[-,-]_\alpha = \alpha \circ [-,-],\alpha),
\]
with $L$ a Lie algebra and $\alpha$ a self Lie algebra morphism, are isomorphic.  Two (multiplicative) Hom-Lie algebras $(L,[-,-],\alpha)$ and $(L',[-,-]',\alpha')$ are said to be \emph{isomorphic} if there is a linear isomorphism $\gamma \colon L \to L'$ such that $\gamma\alpha = \alpha'\gamma$ and $\gamma\circ[-,-] = [-,-]'\circ\gamma^{\otimes 2}$.

\begin{proposition}
\label{prop:HLiso}
Let $\fg$ and $\fh$ be Lie algebras and $\alpha \colon \fg \to \fg$ and $\beta \colon \fh \to \fh$ be Lie algebra morphisms with $\beta$ injective.  Write
\[
\fg_\alpha = (\fg,[-,-]_\alpha = \alpha\circ[-,-],\alpha) \quad\text{and}\quad
\fh_\beta = (\fh,[-,-]_\beta = \beta\circ[-,-],\beta).
\]
Then the following two statements are equivalent:
\begin{enumerate}
\item
The multiplicative Hom-Lie algebras $\fg_\alpha$ and $\fh_\beta$ are isomorphic.
\item
There exists a Lie algebra isomorphism $\gamma \colon \fg \to \fh$ such that $\gamma\alpha = \beta\gamma$.
\end{enumerate}
\end{proposition}

\begin{proof}
As we remarked in the introduction, $\fg_\alpha$ is a multiplicative Hom-Lie algebra whenever $\alpha$ is a Lie algebra morphism on $\fg$, which can be shown by a direct computation \cite{yau2}.  To show that the two statements are equivalent, first suppose that $\fg_\alpha$ and $\fh_\beta$ are isomorphic as multiplicative Hom-Lie algebras.  So there is a linear isomorphism $\gamma \colon \fg \to \fh$ such that $\gamma\alpha = \beta\gamma$ and $[-,-]_\beta\circ \gamma^{\otimes 2} = \gamma \circ [-,-]_\alpha$.   For $x, y \in \fg$, we have
\[
\begin{split}
\beta\gamma[x,y] &= \gamma\alpha[x,y] \\
&= \gamma([x,y]_\alpha) \\
&= [\gamma(x),\gamma(y)]_\beta \\
&= \beta([\gamma(x),\gamma(y)]).
\end{split}
\]
Since $\beta$ is injective, we conclude that
\[
\gamma[x,y] = [\gamma(x),\gamma(y)].
\]
So $\gamma$ is a Lie algebra isomorphism such that $\gamma\alpha = \beta\gamma$.  The converse is proved by essentially the same argument.
\end{proof}

Setting $\fg = \fh$ with $\alpha$ and $\beta$ invertible in Proposition ~\ref{prop:HLiso}, we obtain the following special case.

\begin{corollary}
\label{cor:HLiso}
Let $\alpha, \beta \colon \fg \to \fg$ be Lie algebra automorphisms on a Lie algebra $\fg$.  Then the multiplicative Hom-Lie algebras $\fg_\alpha$ and $\fg_\beta$ are isomorphic if and only if $\alpha$ is conjugate to $\beta$ in the group $\Aut(\fg)$ of Lie algebra automorphisms on $\fg$.
\end{corollary}

Now we consider the three-dimensional Heisenberg algebra $\hei = span_\bC\{X,Y,Z\}$.  Its Lie bracket is determined by
\begin{equation}
\label{eq:heisenberg}
[X,Z] = 0 = [X,Y] \quad\text{and}\quad [Y,Z] = X \quad (\text{the Heisenberg relation}).
\end{equation}
See, e.g., \cite[2.5.8]{hall} or \cite[pp.11-12]{jac}.

\begin{theorem}
\label{thm:H}
Consider the Heisenberg algebra $\hei$ with the basis $\{X,Y,Z\}$ \eqref{eq:heisenberg} and a linear map $\alpha \colon \hei \to \hei$.  Then the following statements hold.
\begin{enumerate}
\item
The map $\alpha$ is a Lie algebra morphism if and only if
\begin{equation}
\label{eq:alphaH}
\alpha = \matrixH
\end{equation}
for some $a_{ij} \in \bC$.  It is invertible if and only if $(a_{22}a_{33} - a_{23}a_{32}) \not= 0$.
\item
Given a Lie algebra morphism $\alpha$ on $\hei$, the Hom-Lie bracket in the corresponding multiplicative Hom-Lie algebra
\[
\hei_\alpha = (\hei,[-,-]_\alpha = \alpha\circ[-,-],\alpha)
\]
is determined by the relations
\[
\begin{split}
[X,Y]_\alpha &= 0 = [X,Z]_\alpha,\\
[Y,Z]_\alpha &= (a_{22}a_{33} - a_{23}a_{32})X \quad (\text{twisted Heisenberg relation}).
\end{split}
\]
\item
There exist uncountably many isomorphism classes of multiplicative Hom-Lie algebras of the form $\hei_\alpha = (\hei,[-,-]_\alpha,\alpha)$.
\end{enumerate}
\end{theorem}

We will omit the proof of the first assertion, since we will prove the much more complicated case for $\fsl(2)$ below.  The second assertion follows from the first one.  The last assertion  in Theorem \ref{thm:H} follows from Corollary \ref{cor:HLiso}.  In fact, the following two Lie algebra automorphisms on $\hei$,
\begin{equation}
\label{eq:alphabeta}
\alpha = \alphaHa \quad\text{and}\quad \beta = \alphaHb,
\end{equation}
are not conjugate in $\Aut(\hei)$, provided $a \not= b^{\pm 1}$.  It follows that their induced multiplicative Hom-Lie algebras, $\hei_\alpha$ and $\hei_\beta$, are not isomorphic.

For example, when $\alpha \colon \hei \to \hei$ has the form \eqref{eq:alphaH}, the solution $B_\alpha$ \eqref{eq:BalphaL} of the HYBE for $(\bC \oplus \hei_\alpha, \alpha)$ satisfies
\[
\begin{split}
B_\alpha\left((a,Y) \otimes (b,Z)\right) &= (b, a_{13}X + a_{23}Y + a_{33}Z) \otimes (a, a_{12}X + a_{22}Y + a_{32}Z)\\
&\relphantom{} + (1,0) \otimes (0,(a_{22}a_{33} - a_{23}a_{32}X)).
\end{split}
\]

Our next example is the three-dimensional Lie algebra $\poincare = span_\bC\{X,Y,Z\}$, whose bracket is determined by the relations
\begin{equation}
\label{eq:poincare}
[Y,Z] = 0, \quad [Y,X] = \frac{1}{2}Y, \quad\text{and}\quad [Z,X] = \frac{1}{2}Z.
\end{equation}
The notation $\poincare$ comes from the fact that this Lie algebra is the linear dual of the Lie coalgebra $\fsl(2)$ \cite[Example 8.1.11]{majid}, which plays a role in the theory of Lie bialgebras and the classical YBE.  Another way to look at it is that $\poincare$ is isomorphic to the $(1+1)$-Poincar\'{e} algebra (see, e.g., \cite{hall} (2.5.9) or \cite{jac} pp.12-13), which is the Lie algebra of the Poincar\'{e} group of affine transformations on $\bR^{2}$ preserving the Lorentz distance.

\begin{theorem}
\label{thm:poincare}
Consider the Lie algebra $\poincare = span_\bC\{X,Y,Z\}$ \eqref{eq:poincare} and a linear map $\alpha \colon \poincare \to \poincare$.  Then the following statements hold.
\begin{enumerate}
\item
The map $\alpha$ is a Lie algebra morphism if and only if its matrix takes the form
\[
\alpha_1 = \alphaPa \quad\text{or}\quad \alpha_2 = \alphaPb,
\]
where $a_{ij} \in \bC$ with $a_{11} \not= 1$.  The map $\alpha_1$ is invertible if and only if $(a_{22}a_{33} - a_{23}a_{32}) \not= 0$.
\item
The Hom-Lie bracket $[-,-]_{\alpha_2} = \alpha_2 \circ [-,-]$ on $\poincare_{\alpha_2}$ is identically zero.  The Hom-Lie bracket $[-,-]_{\alpha_1} = \alpha_1 \circ [-,-]$ on $\poincare_{\alpha_1}$ is determined by the relations
\[
\begin{split}
[Y,Z]_{\alpha_1} &= 0, \\
[Y,X]_{\alpha_1} &= \frac{1}{2}(a_{22}Y + a_{32}Z),\\
[Z,X]_{\alpha_1} &= \frac{1}{2}(a_{23}Y + a_{33}Z).
\end{split}
\]
\item
There are uncountably many isomorphism classes of multiplicative Hom-Lie algebras of the form
\[
\poincare_{\alpha_1} = (\poincare,[-,-]_{\alpha_1},\alpha_1),
\]
where $[-,-]_{\alpha_1} = \alpha_1 \circ [-,-]$.
\end{enumerate}
\end{theorem}

Again we omit the proof of Theorem \ref{thm:poincare}, since it is much easier than the case of $\fsl(2)$.  For examples, the solution $B_{\alpha_1}$ \eqref{eq:BalphaL} of the HYBE for $(\bC \oplus \poincare_{\alpha_1}, \alpha_1)$ satisfies
\[
\begin{split}
B_{\alpha_1}\left((a,Y) \otimes (b,X)\right) 
&= (b, X + a_{21}Y + a_{31}Z) \otimes (a, a_{22}Y + a_{32}Z)\\
&\relphantom{} + (1,0) \otimes \left(0,\frac{1}{2}(a_{22}Y + a_{32}Z)\right).
\end{split}
\]

Our last example of multiplicative Hom-Lie algebras comes from the Lie algebra $\fsl(2) = span_\bC\{X,Y,Z\}$, whose bracket is determined by the relations
\begin{equation}
\label{eq:sl2}
[X,Y] = 2Y,\quad [X,Z] = -2Z,\quad\text{and}\quad [Y,Z] = X.
\end{equation}
It is called the \emph{split three-dimensional simple Lie algebra} in \cite[p.14]{jac}. The Lie algebra $\fsl(2)$ is a crucial example in the structure theory of semisimple Lie algebras, as discussed, for example, in \cite{hum,jac}.  Lie algebra morphisms on $\fsl(2)$ are more complicated than those on the Heisenberg algebra and $\poincare$.

\begin{theorem}
\label{thm:sl2}
The following is a complete list of Lie algebra morphisms $\alpha$ on $\fsl(2)$ with respect to the basis $\{X,Y,Z\}$ \eqref{eq:sl2}.  The induced Hom-Lie bracket $[-,-]_\alpha = \alpha \circ [-,-]$ on the Hom-Lie algebra
\[
\fsl(2)_\alpha = (\fsl(2),[-,-]_\alpha,\alpha)
\]
is also stated in each case.
\begin{enumerate}
\item
$\alpha = 0$, $[-,-]_\alpha = 0$.
\item
\begin{equation}
\label{eq:alpha1sl2}
\alpha_1 = \alphaSa,
\end{equation}
with $a,b,c \in \bC$, $b \not= 0$, and $ac = 0$.
\[
\begin{split}
[X,Y]_{\alpha_1} &= 2(cX + bY - b^{-1}c^2Z), \\
[X,Z]_{\alpha_1} &= -2(aX - a^2bY + b^{-1}Z),\\
[Y,Z]_{\alpha_1} &= X - 2abY - 2b^{-1}cZ.
\end{split}
\]
\item
\begin{equation}
\label{eq:alpha2sl2}
\alpha_2 = \alphaSb,
\end{equation}
with $a,b,c \in \bC$, $b \not= 0$, and $ac = 0$.
\[
\begin{split}
[X,Y]_{\alpha_2} &= 2(cX - b^{-1}c^2Y + bZ), \\
[X,Z]_{\alpha_2} &= -2(aX + b^{-1}Y - a^2bZ),\\
[Y,Z]_{\alpha_2} &= -X + 2b^{-1}cY + 2abZ.
\end{split}
\]
\item
\begin{equation}
\label{eq:alpha3sl2}
\alpha_3 = \alphaSc,
\end{equation}
with $a,b,c \in \bC$, $ab \not= 0$, and $c \not= \pm 1$.
\begin{subequations}
\allowdisplaybreaks
\begin{align*}
[X,Y]_{\alpha_3} &= 2a\left(X + \frac{b}{c-1}Y + \frac{1-c}{b}Z\right),\\
[X,Z]_{\alpha_3} &= \frac{c-1}{2a}\left((c+1)X + bY - \frac{(c+1)^2}{b}Z\right),\\
[Y,Z]_{\alpha_3} &= cX + bY + \frac{1-c^2}{b}Z.
\end{align*}
\end{subequations}
\end{enumerate}
Moreover, $\alpha_1$, $\alpha_2$, and $\alpha_3$ are all invertible.  There are uncountably many isomorphism classes of multiplicative Hom-Lie algebras of the form $sl(2)_{\alpha_*}$.
\end{theorem}

In fact, not only are $\alpha_1$, $\alpha_2$, and $\alpha_3$ invertible, but also their matrices all have determinant equal to $1$.  In \cite[Example 6.1]{yau5}, the special case of $\alpha_1$ with $a = c = 0$ (i.e., when $\alpha_1$ is diagonal) was discussed, together with its induced solution of the HYBE \eqref{eq:BalphaL} for $(\bC \oplus sl(2)_{\alpha_1},\alpha_1)$.

Assume for the moment that the Lie algebra morphisms on $\fsl(2)$ have been classified as in Theorem ~\ref{thm:sl2}.  One computes directly that $\alpha_i$ has determinant $1$ for $i = 1, 2, 3$, and so they are all invertible.  It then follows from Corollary ~\ref{cor:HLiso} that there are uncountably many isomorphism classes of multiplicative Hom-Lie algebras of the form
\[
\fsl(2)_{\alpha_*} = (\fsl(2),\alpha_* \circ [-,-],\alpha_*).
\]
Indeed, using the classification of the Lie algebra morphisms on $\fsl(2)$, one can see that the automorphisms on $\fsl(2)$ represented by the matrices in \eqref{eq:alphabeta} with $a \not= b^{\pm 1}$ are not conjugate in $\Aut(\fsl(2))$.

For example, the solution $B_{\alpha_1}$ \eqref{eq:BalphaL} of the HYBE for $(\bC \oplus \fsl(2)_{\alpha_1}, \alpha_1)$ satisfies
\[
\begin{split}
B_{\alpha_1}\left((a_0,X) \otimes (b_0,Y)\right) 
&= (b_0, cX + bY - b^{-1}c^2Z) \otimes (a_0, X - 2abY - 2b^{-1}cZ)\\
&\relphantom{} + (1,0) \otimes \left(0,2(cX + bY - b^{-1}c^2Z)\right).
\end{split}
\]

%%%%%%%%%%%%%%%%%%%%%%%%%%%%%%%%%%%%%%%%%%%%%%
\begin{proof}[Proof of Theorem ~\ref{thm:sl2}]
First note that the Hom-Lie bracket $[-,-]_\alpha = \alpha\circ[-,-]$ in each case is computed by simply applying $\alpha$ to the Lie bracket in $\fsl(2)$ \eqref{eq:sl2}.  Therefore, it remains to classify the Lie algebra morphisms on $\fsl(2)$.

Let $\alpha \colon \fsl(2) \to \fsl(2)$ be a linear map represented by the $3 \times 3$ matrix $(a_{ij})$ with respect to the basis $\{X,Y,Z\}$ of $\fsl(2)$ \eqref{eq:sl2}.  Then $\alpha$ is a Lie algebra morphism if and only if
\[
\alpha([-,-]) = [\alpha(-),\alpha(-)]
\]
for the nine basis elements of $\fsl(2)^{\otimes 2}$.  By the skew-symmetry of the Lie bracket, we only need to check this equality for the three basis elements $X \otimes Y$, $X \otimes Z$, and $Y \otimes Z$.  The equation
\[
\alpha([X,Y]) = [\alpha(X),\alpha(Y)]
\]
and the defining properties \eqref{eq:sl2} give rise to three equations in the $a_{ij}$, one equation from each of the coefficients of $X$, $Y$, and $Z$.  Likewise, using instead the basis elements $X \otimes Z$ and $Y \otimes Z$, we obtain six more equations in the $a_{ij}$.  A simple calculation shows that these nine equations are:
\begin{subequations}
\allowdisplaybreaks
\begin{align}
a_{21}a_{32} - a_{22}a_{31} &= 2a_{12} \label{sl21}\\
a_{12}a_{21} &= a_{22}(a_{11} - 1) \label{sl22}\\
a_{12}a_{31} &= a_{32}(1 + a_{11}) \label{sl23}\\
a_{21}a_{33} - a_{23}a_{31} &= -2a_{13} \label{sl24}\\
a_{13}a_{21} &= a_{23}(1 + a_{11}) \label{sl25}\\
a_{13}a_{31} &= a_{33}(a_{11} - 1) \label{sl26}\\
a_{11} &= a_{22}a_{33} - a_{23}a_{32} \label{sl27}\\
a_{21} &= 2(a_{12}a_{23} - a_{13}a_{22}) \label{sl28}\\
a_{31} &= 2(a_{13}a_{32} - a_{12}a_{33}) \label{sl29}
\end{align}
\end{subequations}
The first three equations above are from $\alpha([X,Y]) = [\alpha(X),\alpha(Y)]$.  The next three equations are from $\alpha([X,Z]) = [\alpha(X),\alpha(Z)]$.  The last three equations are from $\alpha([Y,Z]) = [\alpha(Y),\alpha(Z)]$.

In the rest of this proof, we show that the solutions to these nine simultaneous equations are exactly the four types of Lie algebra morphisms $\alpha$ stated in Theorem ~\ref{thm:sl2}.  To solve the above simultaneous equations, we consider four cases, depending on whether $a_{12}$ and $a_{13}$ are zero or not.

\textbf{Case I:} $a_{12} = a_{13} = 0$.  We obtain immediately from \eqref{sl28} and \eqref{sl29} that $a_{21} = a_{31} = 0$, which implies $a_{23}a_{33} = 0$ by \eqref{sl25} and \eqref{sl26}.  There are three sub-cases:
\begin{enumerate}
\item
If $a_{23} = 0$ and $a_{33} \not= 0$, then we obtain $a_{11} = 1$ from \eqref{sl26}, $a_{22}a_{33} = 1$ from \eqref{sl27}, and $a_{32} = 0$ from \eqref{sl23}.  So we have
\begin{equation}
\label{eq:sl21}
\alpha = \alphaonea
\end{equation}
with $a_{22} \not= 0$.
\item
If $a_{23} \not= 0$ and $a_{33} = 0$, then we obtain $a_{11} = -1$ from \eqref{sl25}, $a_{23}a_{32} = 1$ from \eqref{sl27}, and $a_{22} = 0$ from \eqref{sl22}.  So we have
\begin{equation}
\label{eq:sl22}
\alpha = \alphaoneb
\end{equation}
with $a_{32} \not= 0$.
\item
If $a_{23} = a_{33} = 0$, then $a_{11} = 0$ from \eqref{sl27}, $a_{32} = 0$ from \eqref{sl23}, and $a_{22} = 0$ from \eqref{sl22}.  So $\alpha = 0$.
\end{enumerate}

\textbf{Case II:} $a_{12} = 0$ and $a_{13} \not= 0$.  We obtain immediately from \eqref{sl22} and \eqref{sl23} that $a_{22}a_{32} = 0$.  We divide this into three sub-cases as in Case I.
\begin{enumerate}
\item
If $a_{22} = 0$ and $a_{32} \not= 0$, then \eqref{sl28} implies $a_{21} = 0$ and \eqref{sl29} implies $a_{31} = 2a_{13}a_{32}$.  Also, $a_{12} = 0$ and $a_{32} \not= 0$ imply by \eqref{sl23} that $a_{11} = -1$.  Now since $a_{21} = 0$, \eqref{sl24} implies $2a_{13} = a_{23}a_{31} = a_{23}(2a_{13}a_{32})$.  Since $a_{13} \not= 0$, we obtain $a_{23}a_{32} = 1$.  Finally, it follows from \eqref{sl26} that $-2a_{33} = a_{13}a_{31} = a_{13}(2a_{13}a_{32})$, so $a_{33} = -a_{13}^2a_{32}$.  In summary, we have
\begin{equation}
\label{eq:sl23}
\alpha = \alphatwoa
\end{equation}
with $a_{32}, a_{13} \not= 0$.
\item
If $a_{22} \not= 0$ and $a_{32} = 0$, then \eqref{sl29} implies $a_{31} = 0$, \eqref{sl28} implies $a_{21} = -2a_{13}a_{22}$, and \eqref{sl22} implies $a_{11} = 1$.  Plugging $a_{11} = 1$ into \eqref{sl27}, we obtain $a_{33} = a_{22}^{-1}$.  Finally, by \eqref{sl25} we have $a_{23} = a_{13}a_{21}/2 = -a_{13}^2a_{22}$.  Therefore, we have
\begin{equation}
\label{eq:sl24}
\alpha = \alphatwob
\end{equation}
with $a_{22}, a_{13} \not= 0$.
\item
Now suppose $a_{22} = a_{32} = 0$.  We show that this sub-case cannot happen.  First we claim that $a_{11} = a_{21} = a_{31} = 0$. Indeed, since $a_{12} = 0 = a_{32}$, \eqref{sl29} implies $a_{31} = 0$.  Likewise, $a_{12} = 0 = a_{22}$ and \eqref{sl28} imply $a_{21} = 0$.  Using \eqref{sl27} we also have $a_{11} = 0$.

Next we claim that $a_{33} = a_{23} = 0$.  Indeed, using $a_{31} = 0 = a_{11}$ and \eqref{sl26}, we obtain $a_{33} = 0$.  Likewise, $a_{21} = 0 = a_{11}$ and \eqref{sl25} imply $a_{23} = 0$.  Therefore, it follows from \eqref{sl24} that $-2a_{13} = 0$, contradicting the assumption $a_{13} \not= 0$.  Thus, this sub-case cannot happen.
\end{enumerate}

\textbf{Case III:} $a_{12} \not= 0$ and $a_{13} = 0$.  We obtain from \eqref{sl25} and \eqref{sl26} that $a_{23}a_{33} = 0$.  Dividing this into three sub-cases as in Case II and performing a similar calculation, we obtain
\begin{equation}
\label{eq:sl25}
\alpha = \alphathreea
\end{equation}
with $a_{12}, a_{22} \not= 0$, or
\begin{equation}
\label{eq:sl26}
\alpha = \alphathreeb
\end{equation}
with $a_{12}, a_{32} \not= 0$.

Notice that the matrices \eqref{eq:sl21}, \eqref{eq:sl24}, and \eqref{eq:sl25} can be stated together as the matrix $\alpha_1$ \eqref{eq:alpha1sl2} in the statement of Theorem ~\ref{thm:sl2}.  Likewise, the matrices \eqref{eq:sl22}, \eqref{eq:sl23}, and \eqref{eq:sl26} can be stated together as the matrix $\alpha_2$ \eqref{eq:alpha2sl2}.

\textbf{Case IV:} $a_{12} \not= 0$ and $a_{13} \not= 0$.  This is the longest of the four cases.  First we claim that
\[
a_{11} \not= \pm 1.
\]
Indeed, if $a_{11} = -1$, then we have $a_{31} = 0$ from \eqref{sl23} and $a_{21} = 0$ from \eqref{sl25}.  This implies that $2a_{12} = 0$ by \eqref{sl21}.  This contradicts the assumption $a_{12} \not = 0$, so $a_{11} \not= -1$.  On the other hand, if $a_{11} = 1$, then $a_{21} = 0$ by \eqref{sl22}.  So \eqref{sl25} implies that $a_{23} = 0$.  But then \eqref{sl24} implies that $a_{13} = 0$, contradicting the assumption $a_{13} \not= 0$.  We have shown that, if $a_{12} \not= 0$ and $a_{13} \not= 0$, then $a_{11} \not= \pm 1$.

Next we claim that
\[
a_{21}a_{22} \not= 0.
\]
Indeed, it follows from \eqref{sl22}, $a_{12} \not= 0$, and $a_{11} \not= 1$ that either $a_{21} = a_{22} = 0$ or both $a_{21}$ and $a_{22}$ are non-zero.  But the former implies by \eqref{sl21} that $2a_{12} = 0$, which is a contradiction.  Thus we must have $a_{21} \not= 0$ and $a_{22} \not= 0$.

The rest of this case is about expressing all the $a_{ij}$ in terms of $a_{11} \not= \pm 1$, $a_{12} \not= 0$, and $a_{21} \not= 0$.  This is a tedious but conceptually elementary calculation.  We begin with $a_{22}$ and $a_{23}$.  First, \eqref{sl22} implies
\begin{equation}
\label{a22}
a_{22} = \frac{a_{12}a_{21}}{a_{11} - 1}.
\end{equation}
Using \eqref{sl22} and \eqref{sl25}, we have
\[
a_{21} = \frac{a_{22}(a_{11} - 1)}{a_{12}} = \frac{a_{23}(a_{11} + 1)}{a_{13}},
\]
from which we obtain
\[
a_{13}a_{22} = a_{12}a_{23} \cdot \frac{a_{11} + 1}{a_{11} - 1}.
\]
Plugging this into \eqref{sl28}, we obtain
\begin{equation}
\label{a21A}
a_{21} = 2a_{12}a_{23}\left(1 - \frac{a_{11} + 1}{a_{11} - 1}\right) = \frac{4a_{12}a_{23}}{1 - a_{11}}.
\end{equation}
Solving for $a_{23}$, we obtain
\begin{equation}
\label{a23}
a_{23} = \frac{a_{21}(1 - a_{11})}{4a_{12}}.
\end{equation}
Also, we infer that $a_{23} \not= 0$, since $a_{21} \not= 0$ and $a_{11} \not= 1$.

Next we consider $a_{13}$ and $a_{31}$.  Using \eqref{sl25} and \eqref{a21A}, we obtain
\begin{equation}
\label{a13}
a_{13} = \frac{a_{23}(1 + a_{11})}{a_{21}} = \frac{1 - a_{11}^2}{4a_{12}}.
\end{equation}
Using \eqref{sl23} and \eqref{a13}, we obtain
\begin{equation}
\label{a31A}
a_{31} = \frac{a_{32}(1 + a_{11})}{a_{12}} = \frac{4a_{13}a_{32}(1 + a_{11})}{1 - a_{11}^2} = \frac{4a_{13}a_{32}}{1 - a_{11}}.
\end{equation}
Plugging the expressions \eqref{a21A} of $a_{21}$, \eqref{a31A} of $a_{31}$, and \eqref{a13} of $a_{13}$ into \eqref{sl24}, we obtain
\[
\frac{4a_{23}(a_{12}a_{33} - a_{13}a_{32})}{1 - a_{11}} = -2a_{13} = \frac{a_{11}^2 - 1}{2a_{12}}.
\]
Combining this with \eqref{sl29} and \eqref{a21A}, we finally obtain
\begin{equation}
\label{a31B}
a_{31} = 2(a_{13}a_{32} - a_{12}a_{33}) = \frac{a_{11}^2 - 1}{2a_{12}} \cdot \frac{1 - a_{11}}{-2a_{23}} = \frac{1 - a_{11}^2}{a_{21}}.
\end{equation}

Now we consider $a_{32}$.  Using again the expressions \eqref{a21A} of $a_{21}$ and \eqref{a31A} of $a_{31}$, it follows from \eqref{sl21} that
\[
2a_{12} = a_{21}a_{32} - a_{22}a_{31} = \frac{4a_{32}(a_{12}a_{23} - a_{13}a_{22})}{1 - a_{11}}.
\]
Solving for $a_{32}$ in the above equation and using \eqref{sl28}, we obtain
\begin{equation}
\label{a32}
a_{32} = \frac{a_{12}(1 - a_{11})}{2(a_{12}a_{23} - a_{13}a_{22})} = \frac{a_{12}(1 - a_{11})}{a_{21}}.
\end{equation}

Finally, we consider $a_{33}$.  Using \eqref{sl26} and \eqref{a13}, we obtain
\begin{equation}
\label{a31C}
a_{31} = \frac{a_{33}(a_{11} - 1)}{a_{13}} = \frac{4a_{12}a_{33}(a_{11} - 1)}{1 - a_{11}^2} = \frac{-4a_{12}a_{33}}{1 + a_{11}}.
\end{equation}
Solving for $a_{33}$ in \eqref{a31C} and using \eqref{a31B}, we obtain
\begin{equation}
\label{a33}
a_{33} = \frac{a_{31}(1 + a_{11})}{-4a_{12}} = \frac{(a_{11}^2 - 1)(a_{11} + 1)}{4a_{12}a_{21}}.
\end{equation}

Now observe that \eqref{a22}, \eqref{a23}, \eqref{a13}, \eqref{a31B}, \eqref{a32}, and \eqref{a33} give exactly the matrix \eqref{eq:alpha3sl2}.  This finishes Case IV and the proof of Theorem ~\ref{thm:sl2}.
\end{proof}

%%%%%%%%%%%%%%%%%%%%%%%%%%%%%%%%%%%%%%%%%%%%%%%%%%%%%%%%%%%%
\section{Solutions of the HYBE from Yetter-Drinfel'd modules}
\label{sec:YD}

The results in this section are about Yetter-Drinfel'd modules and are valid over any field $\bk$ of characteristic $0$.  We will produce a family of solutions of the HYBE from each Yetter-Drinfel'd module.  A \emph{Yetter-Drinfel'd module} $V$ over a bialgebra $H$ \cite{rt,yetter2} consists of
\begin{enumerate}
\item
a left $H$-module structure on $V$, written as $x \cdot v$ for $x \in H$ and $v \in V$, and
\item
a left $H$-comodule structure on $V$, written as $\rho(v) = \sum v_{-1} \otimes v_{0}$ for $v \in V$,
\end{enumerate}
such that the Yetter-Drinfel'd condition
\begin{equation}
\label{eq:YD}
\sum x_1 v_{-1} \otimes x_2 \cdot v_0 = \sum (x_1 \cdot v)_{-1} x_2 \otimes (x_1 \cdot v)_0
\end{equation}
holds for all $x \in H$ and $v \in V$.  Here, and in what follows,
\[
\Delta(x) = \sum x_1 \otimes x_2 \in H^{\otimes 2}
\]
is Sweedler's notation \cite{sweedler} for comultiplication. The coassociativity of the left $H$-comodule structure $\rho$ is the equality
\[
(Id_H \otimes \rho) \circ \rho = (\Delta \otimes Id_V) \circ \rho.
\]
For $u \in V$, we write
\[
(Id_H \otimes \rho)(\rho(u))
= \sum u_{-2} \otimes u_{-1} \otimes u_0
= (\Delta \otimes Id_V)(\rho(u)).
\]
In the special case that $H$ is a finite dimensional Hopf algebra, the category of Yetter-Drinfel'd modules over $H$ is equivalent to the category of left modules over the Drinfel'd  double of $H$ \cite{dri87,majid91a}.

Using the Yetter-Drinfel'd condition \eqref{eq:YD}, a direct computation shows that each Yetter-Drinfel'd module $V$ over a bialgebra $H$ has an associated $R$-matrix \cite{lr,rad} given by
\begin{equation}
\label{YDB}
B(v \otimes w) = \sum v_{-1} \cdot w \otimes v_0
\end{equation}
for $v, w \in V$.  The following result shows that this operator $B$ is also a solution of the HYBE for $(V,\alpha)$, provided that $\alpha$ is compatible with the $H$-(co)module structures.

\begin{theorem}
\label{thm:YD}
Let $V$ be a Yetter-Drinfel'd module over a bialgebra $H$, and let $\alpha \colon V \to V$ be a linear map that is both a morphism of $H$-modules and a morphism of $H$-comodules.  Then $B$ in \eqref{YDB} is a solution of the HYBE \eqref{eq:HYBE} for $(V,\alpha)$.
\end{theorem}

\begin{proof}
With $B$ defined as in \eqref{YDB}, it follows from the $H$-(co)linearity assumption on $\alpha$ that $B$ commutes with $\alpha^{\otimes 2}$.

The assumption that $\alpha \colon V \to V$ is an $H$-comodule morphism means that
\[
\rho(\alpha(v)) = (Id_H \otimes \alpha)(\rho(v)),
\]
i.e.,
\[
\sum \alpha(v)_{-1} \otimes \alpha(v)_0 = \sum v_{-1} \otimes \alpha(v_0).
\]
To simplify typography, we will omit the summation signs in $\Delta(x)$ (for $x \in H$) and in $\rho(v)$ (for $v \in V$).  Let $\gamma$ denote a typical generator $u \otimes v \otimes w \in V^{\otimes 3}$.  The left-hand side of the HYBE \eqref{eq:HYBE}, when applied to $\gamma$, is
\begin{subequations}
\allowdisplaybreaks
\begin{align*}
(\alpha \otimes B) \circ (B \otimes \alpha) \circ (\alpha \otimes B)(\gamma)
&= (\alpha \otimes B) \circ (B \otimes \alpha)\left(\alpha(u) \otimes v_{-1} \cdot w \otimes v_0\right)\\
&= (\alpha \otimes B)\left(u_{-1} \cdot (v_{-1} \cdot w) \otimes \alpha(u_0) \otimes \alpha(v_0)\right)\\
&= \left(u_{-2}v_{-1}\right) \cdot \alpha(w) \otimes u_{-1} \cdot \alpha(v_0) \otimes \alpha(u_0)\\
&= \left(u_{-2}\alpha(v)_{-1}\right)\cdot\alpha(w) \otimes u_{-1} \cdot \alpha(v)_0 \otimes \alpha(u_0).
\end{align*}
\end{subequations}
In the third equality above, we used
\[
\begin{split}
\alpha(u_{-1} \cdot (v_{-1} \cdot w))
&= \alpha((u_{-1}v_{-1}) \cdot w)\\
&= \left(u_{-1}v_{-1}\right) \cdot \alpha(w).
\end{split}
\]
Likewise, the right-hand side of the HYBE, when applied to $\gamma$, is
\begin{subequations}
\allowdisplaybreaks
\begin{align*}
(B \otimes \alpha) \circ (\alpha \otimes B) \circ (B \otimes \alpha)(\gamma)
&= (B \otimes \alpha) \circ (\alpha \otimes B)\left(u_{-1} \cdot v \otimes u_0 \otimes \alpha(w)\right)\\
&= (B \otimes \alpha)\left(u_{-2} \cdot \alpha(v) \otimes u_{-1} \cdot \alpha(w) \otimes u_0\right)\\
&= \left(u_{-2} \cdot \alpha(v)\right)_{-1} \cdot \left(u_{-1} \cdot \alpha(w)\right) \otimes \left(u_{-2} \cdot \alpha(v)\right)_0 \otimes \alpha(u_0)\\
&= \left[\left(u_{-2} \cdot \alpha(v)\right)_{-1} u_{-1}\right] \cdot \alpha(w) \otimes \left(u_{-2} \cdot \alpha(v)\right)_0 \otimes \alpha(u_0).
\end{align*}
\end{subequations}
The Yetter-Drinfel'd condition \eqref{eq:YD} (with $u_{-1}$ and $\alpha(v)$ instead of $x$ and $v$) now implies that $B$ \eqref{YDB} is a solution of the HYBE for $(V,\alpha)$.
\end{proof}

Let us discuss two special cases of Theorem ~\ref{thm:YD} that are closely related to quantum groups.  These results were first obtained in \cite{yau5}.  A \emph{quasi-triangular bialgebra} $(H,R)$ \cite{dri87,dri89} consists of a bialgebra $H$ and an invertible element $R \in H^{\otimes 2}$ such that
\[
\begin{split}
\tau(\Delta(x)) &= R \Delta(x) R^{-1},\\
(\Delta \otimes Id)(R) &= R_{13}R_{23},\\
(Id \otimes \Delta)(R) &= R_{13}R_{12}
\end{split}
\]
for $x \in H$, where $\tau$ interchanges the two tensor factors.  Here if $R = \sum s_i \otimes t_i \in H^{\otimes 2}$, then
\[
\begin{split}
R_{13}R_{23} &= \sum s_i \otimes s_j \otimes t_it_j,\\
R_{13}R_{12} &= \sum s_is_j \otimes t_j \otimes t_i.
\end{split}
\]
For example, a unital cocommutative bialgebra is quasi-triangular with $R = 1 \otimes 1$.  However, most interesting quasi-triangular bialgebras, including many quantum groups, are not cocommutative.

Let $V$ be an $H$-module for some quasi-triangular bialgebra $(H,R)$, where $R = \sum s_i \otimes t_i \in H^{\otimes 2}$.  Then $V$ becomes a Yetter-Drinfel'd module over $H$, in which the left $H$-comodule structure is
\[
\rho(v) = \sum t_i \otimes s_i \cdot v.
\]
Indeed, that $\rho$ gives $V$ a left $H$-comodule structure is a consequence of the assumption $(Id \otimes \Delta)(R) = R_{13}R_{12}$, and the Yetter-Drinfel'd condition \eqref{eq:YD} follows from $(\tau(\Delta(x)))R = R\Delta(x)$.   For this Yetter-Drinfel'd module structure on $V$, the $R$-matrix $B$ in \eqref{YDB} takes the form
\[
B(v \otimes w) = \sum t_i \cdot w \otimes s_i \cdot v = \tau(R(v \otimes w)).
\]
If $ \alpha \colon V \to V$ is an $H$-module morphism, then it is also an $H$-comodule morphism, since $\alpha (s_i \cdot v) = s_i \cdot \alpha(v)$.  Therefore, we recover the following result from \cite[Theorem 1.2]{yau5} as a special case of Theorem ~\ref{thm:YD}.

\begin{corollary}[\cite{yau5}]
\label{cor:YD}
Let $V$ be an $H$-module for some quasi-triangular bialgebra $(H,R)$, and let $\alpha \colon V \to V$ be an $H$-module morphism.  Then
\[
B = \tau \circ R
\]
is a solution of the HYBE \eqref{eq:HYBE} for $(V,\alpha)$.
\end{corollary}

Another special case of Theorem ~\ref{thm:YD} concerns the dual of Corollary ~\ref{cor:YD}.  A \emph{dual quasi-triangular bialgebra} $(H,R)$ \cite{hay,kassel,lt,majid91,majid,sch} consists of a bialgebra $H$ and an invertible bilinear form $R$ on $H$ under the convolution product in $\Hom(H^{\otimes 2},\bk)$, such that the following conditions hold for $x,y,z \in H$:
\begin{equation}
\label{dqb}
\begin{split}
\sum y_1x_1R(x_2 \otimes y_2) &= \sum R(x_1 \otimes y_1)x_2y_2,\\
R(xy \otimes z) & = \sum R(x \otimes z_1)R(y \otimes z_2),\\
R(x \otimes yz) & = \sum R(x_1 \otimes z)R(x_2 \otimes y).
\end{split}
\end{equation}
Let $(H,R)$ be a dual quasi-triangular bialgebra and $V$ be a left $H$-comodule with structure map $\rho(v) = \sum v_{-1} \otimes v_0$.  Then $V$ has a left $H$-module structure $\lambda \colon H \otimes V \to V$ given by
\[
\lambda(x \otimes v) = x \cdot v = R(v_{-1} \otimes x)v_0
\]
for $x \in H$ and $v \in V$.  The fact that $\lambda$ gives $V$ the structure of a left $H$-module is a consequence of the last condition in \eqref{dqb}.  Equipped with the left $H$-module structure $\lambda$ and the left $H$-comodule structure $\rho$, $V$ becomes a Yetter-Drinfel'd module over $H$, in which the Yetter-Drinfel'd condition \eqref{eq:YD} is a consequence of the first condition in \eqref{dqb}.  For this Yetter-Drinfel'd module structure on $V$, the $R$-matrix $B$ in \eqref{YDB} takes the form
\begin{equation}
\label{Bdqb}
B(v \otimes w) = \sum v_{-1} \cdot w \otimes v_0 = \sum R(w_{-1}\otimes v_{-1}) w_0 \otimes v_0.
\end{equation}
If $\alpha \colon V \to V$ is a morphism of $H$-comodules, then it is also a morphism of $H$-modules because $\sum v_{-1} \otimes \alpha(v_0) = \sum \alpha(v)_{-1} \otimes \alpha(v)_0$.  Therefore, we recover the following result from \cite[Theorem 1.3]{yau5} as a special case of Theorem ~\ref{thm:YD}.

\begin{corollary}[\cite{yau5}]
\label{cor2:YD}
Let $V$ be an $H$-comodule for some dual quasi-triangular bialgebra $(H,R)$, and let $\alpha \colon V \to V$ be an $H$-comodule morphism.  Then $B$ in \eqref{Bdqb} is a solution of the HYBE \eqref{eq:HYBE} for $(V,\alpha)$.
\end{corollary}

It is worth noting that obtaining Corollary ~\ref{cor2:YD} as a consequence of Theorem ~\ref{thm:YD} is much simpler than proving it directly, as was done in \cite[section 5]{yau5}.

%%%%%%%%%%%%%%%%%%%%%%%%%%%%%%%%%%%%%%%%%%%%%%%%%%%%%%%%%%%
\section{Solutions of the HYBE on $n$-fold tensor products}
\label{sec:Vn}

Our final result is about constructing an infinite sequence of new solutions of the HYBE from a given one.  To describe this result, recall the generators $\sigma_i$ $(1 \leq i \leq n-1)$ in the braid group $\br$ on $n$ strands, as discussed in the paragraph containing \eqref{eq:braidrelations}.  Let $\Sigma_n$ be the symmetric group on $n$ letters.  The \emph{length} $l(\gamma)$ of a permutation $\gamma \in \Sigma_n$ is the least integer $l$ so that $\gamma$ can be decomposed into a product of $l$ transpositions $\tau_i = (i,i+1)$ $(1 \leq i \leq n-1)$.  The length $l(\gamma)$ is also equal to the number of pairs $i<j$ for which $\gamma(i) > \gamma(j)$.  A decomposition
\[
\gamma = \gamma_1 \cdots \gamma_k
\]
with each $\gamma_i \in \Sigma_n$ is called \emph{reduced} if
\[
l(\gamma) = l(\gamma_1) + \cdots + l(\gamma_k).
\]
By a well-known result of Iwahori \cite{iwahori}, there is a well-defined map $\theta \colon \Sigma_n \to \br$ given by
\[
\theta(\gamma) = \sigma_{i_1} \cdots \sigma_{i_{l(\gamma)}},
\]
where
\[
\gamma = \tau_{i_1} \cdots \tau_{i_{l(\gamma)}}
\]
is any reduced decomposition of $\gamma$ in terms of transpositions.

Now let $B \colon V^{\otimes 2} \to V^{\otimes 2}$ be a solution of the HYBE for $(V,\alpha)$ in which both $\alpha$ and $B$ are invertible.  As discussed in the Introduction \cite[Theorem 1.4]{yau5}, there is a braid group representation $\rho^B_n \colon \br \to \Aut(V^{\otimes n})$ determined by
\[
\rho^B_n(\sigma_i) = B_i,
\]
where
\[
B_i = \alpha^{\otimes (i-1)} \otimes B \otimes \alpha^{\otimes (n-i-1)}
\]
as in \eqref{eq:Bi}.  Pre-composing this map with $\theta \colon \Sigma_n \to \br$ from the previous paragraph, we have a well-defined map
\[
\rho^B_n \theta \colon \Sigma_n \to \Aut(V^{\otimes n}).
\]
We write
\[
\rho^B_n \theta(\gamma) = B^\gamma \in \Aut(V^{\otimes n})
\]
for $\gamma \in \Sigma_n$.  For positive integers $i$ and $j$, define the permutation $\chi_{ij} \in \Sigma_{i+j}$ as
\begin{equation}
\label{chi}
\chi_{ij} =
\begin{pmatrix}
1 & \cdots & i & i+1 & \cdots & i+j\\
j+1 & \cdots & j+i & 1 & \cdots & j
\end{pmatrix}.
\end{equation}
Finally, define the map
\[
\alpha_n = (\alpha^{\otimes n})^{n^2} \colon V^{\otimes n} \to V^{\otimes n},
\]
i.e., the composition of $n^2$ copies of $\alpha^{\otimes n}$.  Then we have the following result, which gives a solution of the HYBE on the $n$-fold tensor product.

\begin{theorem}
\label{thm:Vn}
Let $B \colon V^{\otimes 2} \to V^{\otimes 2}$ be a solution of the HYBE for $(V,\alpha)$ in which both $\alpha$ and $B$ are invertible, and let $n$ be a positive integer.  Then
\[
B^{\chi_{nn}} \colon V^{\otimes n} \otimes V^{\otimes n} \to V^{\otimes n} \otimes V^{\otimes n}
\]
is a solution of the HYBE for $(V^{\otimes n},\alpha_n)$.  Moreover, both $\alpha_n$ and $B^{\chi_{nn}}$ are invertible.
\end{theorem}

Theorem ~\ref{thm:Vn} is a generalization of \cite[Proposition 1.1]{hh}, which deals with $R$-matrices on $n$-fold tensor products.

For example, when $n=1$ we have
\[
B^{\chi_{11}} = B
\]
because $\chi_{11} = (1,2)$ and $\alpha_1 = \alpha$.  So the case $n=1$ of Theorem ~\ref{thm:Vn} is just restating the assumption that $B$ is a solution of the HYBE for $(V,\alpha)$.  When $n=2$, we have $\alpha_2 = (\alpha^{\otimes 2})^4$.  Since
\[
\chi_{22} = (2,3)(3,4)(1,2)(2,3)
\]
is a reduced decomposition of $\chi_{22}$, we have
\[
B^{\chi_{22}} = B_2 \circ B_3 \circ B_1 \circ B_2 \colon (V^{\otimes 2})^{\otimes 2} \to (V^{\otimes 2})^{\otimes 2},
\]
where $B_1 = B \otimes \alpha^{\otimes 2}$, $B_2 = \alpha \otimes B \otimes \alpha$, and $B_3 = \alpha^{\otimes 2} \otimes B$.

\begin{proof}[Proof of Theorem ~\ref{thm:Vn}]
Let us first observe that $B^{\chi_{nn}}$ commutes with $\alpha_n^{\otimes 2}$.  Note that the length of $\chi_{nn}$ \eqref{chi} is $n^2$.  By the definitions of the maps $\theta$ and $\rho^B_n$, we see that $B^{\chi_{nn}}$ is the composition of $n^2$ operators
\[
B_i \colon V^{\otimes n} \otimes V^{\otimes n} \to V^{\otimes n} \otimes V^{\otimes n}
\]
for $1 \leq i \leq 2n-1$, which are defined in \eqref{eq:Bi}.  Since each $B_i$ commutes with $\alpha^{\otimes 2n}$, so does $B^{\chi_{nn}}$.  Therefore, $B^{\chi_{nn}}$ also commutes with $(\alpha^{\otimes 2n})^{n^2} = \alpha_n^{\otimes 2}$.

Since each $B_i$ \eqref{eq:Bi} is invertible, so is $B^{\chi_{nn}}$.  Also, since $\alpha$ is invertible, so are $\alpha^{\otimes n}$ and $\alpha_n = (\alpha^{\otimes n})^{n^2}$.

It remains to check the HYBE \eqref{eq:HYBE} for $B^{\chi_{nn}}$ and $\alpha_n$.  Consider the maps $B_i$ \eqref{eq:Bi} on $V^{\otimes 2n}$.  We have
\[
\alpha^{\otimes n} \otimes B_i = \alpha^{\otimes (n+i-1)} \otimes B \otimes \alpha^{\otimes (2n-i-1)},
\]
which is $B_{n+i}$ on $V^{\otimes 3n}$.  Since $B^{\chi_{nn}}$ is the composition of $n^2$ operators $B_i$ ($1 \leq i \leq 2n-1$), it follows that
\[
\begin{split}
\alpha_n \otimes B^{\chi_{nn}}
&= (\alpha^{\otimes n})^{n^2} \otimes B^{\chi_{nn}}\\
&= B^{1_n \times \chi_{nn}},
\end{split}
\]
where $1_n$ is the identity in $\Sigma_n$ and $1_n \times \chi_{nn} \in \Sigma_{3n}$.  Likewise, we have
\[
B^{\chi_{nn}} \otimes \alpha_n = B^{\chi_{nn} \times 1_n}.
\]
It follows that the two sides of the HYBE are:
\[
\begin{split}
(\alpha_n \otimes B^{\chi_{nn}}) \circ (B^{\chi_{nn}} \otimes \alpha_n) \circ (\alpha_n \otimes B^{\chi_{nn}}) &= B^{1_n \times \chi_{nn}} \circ B^{\chi_{nn} \times 1_n} \circ B^{1_n \times \chi_{nn}},\\
(B^{\chi_{nn}} \otimes \alpha_n) \circ (\alpha_n \otimes B^{\chi_{nn}}) \circ (B^{\chi_{nn}} \otimes \alpha_n) &= B^{\chi_{nn} \times 1_n} \circ B^{1_n \times \chi_{nn}} \circ B^{\chi_{nn} \times 1_n}.
\end{split}
\]
In $\Sigma_{3n}$, both
\[
(1_n \times \chi_{nn})(\chi_{nn} \times 1_n)(1_n \times \chi_{nn})\quad\text{and}\quad
(\chi_{nn} \times 1_n)(1_n \times \chi_{nn})(\chi_{nn} \times 1_n)
\]
are reduced decompositions of the permutation
\[
\gamma =
\begin{pmatrix}
1 & \cdots & n & n+1 & \cdots & 2n & 2n+1 & \cdots & 3n\\
2n+1 & \cdots & 3n & n+1 & \cdots & 2n & 1 & \cdots & n
\end{pmatrix},
\]
since $\chi_{nn} \times 1_n$ and $1_n \times \chi_{nn}$ both have length $n^2$ and $\gamma$ has length $3n^2$.  Therefore, both sides of the HYBE are equal to $B^{\gamma}$.  This finishes the proof that $B^{\chi_{nn}}$ is a solution of the HYBE \eqref{eq:HYBE} for $(V^{\otimes n},\alpha_n)$.
\end{proof}

%%==============%%
%%              %%
%%  References  %%
%%              %%
%%==============%%

\end{document}